\documentclass[journal]{IEEEtran}
\usepackage{amsmath}

\usepackage[hidelinks]{hyperref}
\usepackage{lineno}
\usepackage{subcaption}
\usepackage{hyperref}
\usepackage{cite}
\usepackage{amsmath,amssymb,amsfonts, cases} 
\usepackage{amsmath}
\usepackage{amsthm}

\usepackage{algorithmic}
\usepackage{graphicx}
\usepackage{textcomp}
\usepackage{xcolor}
\usepackage{graphicx}
\usepackage{float}
\usepackage{subcaption}
\usepackage{amsmath,mleftright}
\usepackage{amsfonts,amssymb}
\usepackage{mathrsfs}
\usepackage{mathtools}
\usepackage{algorithm}
\usepackage{algorithmic}
\usepackage{bm}
\usepackage{multirow}
\usepackage{array}
\usepackage{amssymb}
\usepackage{amsmath}
\usepackage{cite}
\usepackage{url}
\usepackage{xcolor}
\usepackage{cite,graphicx,amsmath,amssymb}
\usepackage{fancyhdr}
\usepackage{mdwmath}
\usepackage{mdwtab}
\usepackage{caption}
\usepackage{amsthm}
\usepackage{setspace}
\usepackage{bm}
\usepackage{mathtools}
\usepackage{dsfont}
\usepackage{bbm}
\usepackage{framed}

\newtheorem{theorem}{Theorem}

\newtheorem{lemma}{Lemma}

\newtheorem{corollary}{Corollary}

\makeatletter
\newcommand{\biggg}{\bBigg@{3}}
\newcommand{\Biggg}{\bBigg@{3.5}}
\makeatother
\makeatletter
\renewcommand{\maketag@@@}[1]{\hbox{\m@th\normalsize\normalfont#1}}%
\makeatother
\def\BibTeX{{\rm B\kern-.05em{\sc i\kern-.025em b}\kern-.08em
    T\kern-.1667em\lower.7ex\hbox{E}\kern-.125emX}}
    \expandafter\def\expandafter\normalsize\expandafter{%
    \normalsize%
    \setlength\abovedisplayskip{4pt}%
    \setlength\belowdisplayskip{4pt}%
    \setlength\abovedisplayshortskip{2pt}%
    \setlength\belowdisplayshortskip{2pt}%
}
\allowdisplaybreaks[4]
\begin{document}
\title{Optimal Analog Beamforming and Power Allocation for Multiuser TDMA Systems}
\author{Songnan~Gu, Chongjun~Ouyang, Hao~Jiang, and Xingqi Zhang\vspace{-10pt}
\thanks{S. Gu, C. Ouyang, and H. Jiang are with the School of Electronic Engineering and Computer Science, Queen Mary University of London, London, E1 4NS, U.K. (e-mail: s.gu@se23.qmul.ac.uk; \{c.ouyang, hao.jiang\}@qmul.ac.uk).} 
\thanks{X. Zhang is with Department of Electrical and Computer Engineering, University of Alberta, Edmonton AB, T6G 2R3, Canada (email: xingqi.zhang@ualberta.ca).}
}
\maketitle
\begin{abstract}
The joint design of analog beamforming and power allocation is investigated for a single radio-frequency chain multiuser time-division multiple access system under a max-min signal-to-noise ratio (SNR) criterion. A hardware-efficient phased-array architecture is considered, where the beamforming vector is shared by all users and is subject to constant-modulus constraints. For any fixed analog beamformer, the optimal power allocation is first derived in closed form, by which the original problem is reduced to phase-shift optimization only. Then, globally optimal branch-and-bound (BB) algorithms are developed for discrete and continuous phase shifts. Numerical results show that the proposed BB algorithms achieve the global optimum and provide reliable benchmarks for evaluating the performance gap of low-complexity alternating-optimization methods.
\end{abstract}
\section{Introduction}
Analog beamforming with phased arrays is a hardware-efficient transmission architecture for multi-antenna systems. It has received particular attention in millimeter-wave communications, where directional beamforming is essential to compensate for severe propagation loss \cite{heath2016overview}. More generally, it is also attractive in practical deployments that require low radio-frequency (RF)-chain cost and power consumption, such as fixed wireless access and small-cell backhaul \cite{pi2016millimeter}. These considerations make the analog beamforming architecture a meaningful architecture for low-cost multiuser transmission.

Among various implementations, the single-RF-chain architecture is especially attractive due to its low hardware complexity. In this architecture, one RF chain is connected to all antennas through analog phase shifters, which imposes constant-modulus constraints on the beamforming weights and makes the beam design problem highly non-convex. In multiuser systems with only one RF chain, time-division multiple access (TDMA) is a natural transmission strategy, since only one data stream can be transmitted at each time instant. Under this setting, analog beamforming and power allocation should be designed jointly to improve user fairness.

A large body of existing work has studied phase-shift design under constant-modulus constraints. Representative approaches include semidefinite relaxation, majorization-minimization, manifold optimization, and alternating optimization (AO) \cite{luo2010sdr,yu2016alternating,arora2021efficient,wu2024analog}. Although these methods often achieve good performance with low complexity, they generally provide heuristic or locally optimal solutions. Consequently, the global-optimal design for single-RF-chain multiuser TDMA systems remains much less understood, and the optimality loss of practical low-complexity methods is difficult to quantify. A related problem was considered in \cite{xiao2019user}, where an AO-based method was developed, but a globally optimal benchmark is still unavailable.

To address this issue, this letter studies the joint design of analog beamforming and power allocation for a single-RF-chain multiuser TDMA system under a max-min fairness criterion. We first derive the optimal power allocation in closed form for any fixed beamformer. Based on this result, the original problem is reduced to phase-shift optimization only. We then develop globally optimal branch-and-bound (BB) algorithms \cite{narendra1977branch} for three representative phase-shift settings, namely binary phase shifts, general $M$-ary discrete phase shifts, and continuous phase shifts. Numerical results show that the proposed BB algorithms provide the global optimum and serve as reliable benchmarks for evaluating the performance gap of low-complexity AO-based designs.

\section{System Model}
We consider a downlink single-RF-chain phased-array system with $N$ antennas serving $K$ users. Since only one RF chain is available, the transmitter employs analog beamforming only. The users are served sequentially according to a TDMA protocol. Let $\mathcal{K} \triangleq \{1, \ldots, K\}$ denote the user set. The transmitted signal intended for user $k\in{\mathcal{K}}$ is $\mathbf{x}_k = \sqrt{P_k}\mathbf{v}s_k$, where $P_k$ is the transmit power allocated to user $k$, and $s_k$ denotes the normalized information symbol with $\mathbb{E}[s_k] = 0$ and $\mathbb{E}[|s_k|^2] = 1$. The analog beamforming vector $\mathbf{v} \in \mathbb{C}^{N \times 1}$ is given by \cite{heath2016overview}
\begin{equation}
\mathbf{v} = \frac{1}{\sqrt{N}} \left[ e^{-j\theta_1}, \ldots, e^{-j\theta_N} \right]^T,
\end{equation}
where $\theta_n \in [0, 2\pi)$ denotes the phase shift introduced by the phase shifter connected to the $n$th antenna, for $n = 1, \ldots, N$. 

The received signal at user $k$ is given by
\begin{equation}
y_k = \mathbf{h}_k^H \mathbf{x}_k + n_k = \sqrt{P_k}\mathbf{h}_k^H \mathbf{v}s_k + n_k,
\end{equation}
where $\mathbf{h}_k \in \mathbb{C}^{N \times 1}$ denotes the channel vector from the phased array to user $k$, and $n_k \sim \mathcal{CN}(0, \sigma_k^2)$ is the additive white Gaussian noise with variance $\sigma_k^2$. Accordingly, the received signal-to-noise ratio (SNR) at user $k$ is $\gamma_k = \frac{P_k}{\sigma_k^2} \left| \mathbf{h}_k^H \mathbf{v} \right|^2$. For simplicity, we assume identical noise power across all users, that is, $\sigma_k^2 = \sigma^2$ for all $k\in{\mathcal{K}}$.

To ensure fairness, we maximize the minimum received SNR among all users. The resulting optimization problem is formulated as follows:
\begin{subequations}
\begin{align}
\max_{\boldsymbol{\theta}, \{P_k\}} \quad
& \min_{k \in \mathcal{K}} \gamma_k(\boldsymbol{\theta}, P_k) \\
\text{s.t.} \quad
& \sum_{k=1}^{K} P_k \le P, \label{Per_User_Power_Cons1}\\
& P_k \ge 0, \quad \forall k \in \mathcal{K}, \label{Per_User_Power_Cons2}\\
& \theta_n \in \mathcal{F}, \quad \forall n = 1, \ldots, N,\label{Phase_Shifter_Cons}
\end{align}
\end{subequations}
where $P$ denotes the total transmit power budget, $\boldsymbol{\theta}=[\theta_1,\ldots,\theta_N]^T$, and $\mathcal{F}$ denotes the feasible phase-shift set. For example, $\mathcal{F} = [0, 2\pi)$ corresponds to continuous-phase beamforming. By introducing an auxiliary variable $t$, the above max-min problem can be equivalently rewritten in epigraph form as follows:
\begin{subequations}
\begin{align}
\min_{\boldsymbol{\theta}, \{P_k\}, t} \quad
& -t \\
\text{s.t.} \quad
& \gamma_k \ge t, \quad \forall k \in \mathcal{K}, \\
& \eqref{Per_User_Power_Cons1}, \eqref{Per_User_Power_Cons2}, \eqref{Phase_Shifter_Cons}.
\end{align}
\end{subequations}

The above problem is a non-convex optimization problem. The difficulty comes from two aspects. The power variables and the phase shifts are coupled in the objective, and the analog beamformer must satisfy the unit-modulus constraint. In what follows, we first derive the optimal power allocation for a fixed beamforming vector by using convex optimization. This step reduces the original joint design problem to a phase-shift optimization problem. We then optimize the analog beamformer for several typical feasible phase-shift sets.

\section{Optimal Power Allocation for Fixed Beamforming}
For a given beamforming vector $\mathbf{v}$, the subproblem with respect to $\{P_k\}$ and $t$ is convex. Define the effective channel gain of user $k$ as $G_k \triangleq \frac{1}{N\sigma^2} \left| \mathbf{h}_k^H \mathbf{w} \right|^2>0$ with $\mathbf{w} \triangleq  \left[ e^{-j\theta_1}, \dots, e^{-j\theta_N} \right]^T$. The Lagrangian function can be written as follows \cite{chi2017convex}:
\begin{equation}
\mathcal{L} = -t + \sum_{k=1}^K \lambda_k (t - P_k G_k)
+ \mu \left( \sum_{k=1}^K P_k - P \right)
- \sum_{k=1}^K \eta_k P_k,
\end{equation}
where $\lambda_k \ge 0$, $\mu \ge 0$, and $\eta_k \ge 0$ are the Lagrange multipliers associated with the SNR constraints, the total power constraint, and the non-negativity constraints, respectively.

The Karush-Kuhn-Tucker conditions for the optimal solution $\{P_k^*\}$ and $ t^*$ are given by
\begin{subequations}
\begin{align}
&\frac{\partial \mathcal{L}}{\partial t} = -1 + \sum_{k=1}^K \lambda_k = 0, \label{kkt:t} \\
&\frac{\partial \mathcal{L}}{\partial P_k} = -\lambda_k G_k + \mu - \eta_k = 0, \quad \forall k \in \mathcal{K}, \label{kkt:pk} \\
&\lambda_k (t - P_k G_k) = 0, \quad \forall k \in \mathcal{K}, \label{kkt:slack1} \\
&\mu \left( \sum_{k=1}^K P_k - P \right) = 0, \label{kkt:slack2} \\
&\eta_k P_k = 0, \quad \forall k \in \mathcal{K}. \label{kkt:slack3}
\end{align}
\end{subequations}
From the stationarity condition with respect to $t$, we have $\sum_k \lambda_k = 1$. Hence, at least one $\lambda_k$ is strictly positive. By complementary slackness, this gives $t = P_k G_k$ for every user with $\lambda_k > 0$. Since the objective is to maximize the minimum SNR, all users must achieve the same SNR at the optimum. Otherwise, power can be reallocated from a user with a larger SNR to a user with a smaller SNR and thereby increase the minimum SNR. Therefore, $t=P_kG_k$, $\forall k\in{\mathcal{K}}$. Moreover, the total power budget must be fully used at the optimum, which yields $\sum_{k=1}^K P_k = P$. Substituting $P_K=\frac{t}{G_k}$ into $\sum_{k=1}^K P_k = P$ gives
\begin{equation}
t^* = \frac{P}{\sum_{k=1}^K 1/G_k}, \quad
P_k^* = \frac{P/G_k}{\sum_{j=1}^K 1/G_j}, \quad \forall k \in \mathcal{K}.
\end{equation}
This leads to the following theorem.
\vspace{-5pt}
\begin{theorem}
For any fixed feasible beamforming vector $\mathbf{w}$, the optimal power allocation is given by the above closed-form expression. The corresponding maximum minimum SNR is $t^*$, and all users achieve the same received SNR at the optimum, namely, $\gamma_k(\boldsymbol{\theta}, P_k^*) = t^*$ holds for all $k \in \mathcal{K}$.
\end{theorem}
\vspace{-5pt}
\begin{IEEEproof}
Since the subproblem with fixed $\mathbf{w}$ is convex, the KKT conditions are sufficient for optimality. The stationarity condition gives $\sum_k \lambda_k = 1$, and thus at least one SNR constraint is active. At the optimum, all users must have the same received SNR. Otherwise, one can increase the minimum SNR through power reallocation. Combining $t = P_k G_k$ for $k$ with $\sum_k P_k = P$ yields $\gamma_k(\boldsymbol{\theta}, P_k^*) = t^*$ for all $k$.
\end{IEEEproof}
By substituting the optimal power allocation into the original problem and removing the constant factors $P$ and $\sigma^2$, the joint design problem reduces to the following:
\begin{equation}
\label{eq:sum_inverse_gain}
\begin{aligned}
\min_{\boldsymbol{\theta}} \quad & \sum_{k=1}^{K} \frac{1}{\left| \mathbf{h}_k^H \mathbf{w} \right|^2} \\
\text{s.t.} \quad & \theta_n \in \mathcal{F}, \quad \forall n = 1, \dots, N.
\end{aligned}
\end{equation}
In the following, we focus on the analog beamforming design for problem \eqref{eq:sum_inverse_gain}.
\section{Phase-Shift Design}
Problem \eqref{eq:sum_inverse_gain} is non-convex, and a closed-form optimal solution is difficult to obtain. A natural baseline is an AO method that updates the $N$ phase shifts sequentially in an element-wise manner. However, this approach may get trapped in local optima. To address this issue, we develop BB algorithms under three representative phase-shift constraints, and use numerical results to quantify the performance gap between the AO-based method and the proposed globally optimal BB designs.

Specifically, we consider the following cases. In the binary phase-shift case, $\mathcal{F}=\{0,\pi\}$, which is typical for low-cost phased-array implementations. In the $M$-ary phase-shift case, where $\mathcal{F} = \{2\pi m/M\}_{m=0}^{M-1}$, the model captures practical finite-resolution phase shifters. In the continuous phase-shift case, $\mathcal{F}=[0, 2\pi)$. This case serves as an ideal benchmark, although it usually incurs a higher implementation cost.

\subsection{Binary-Phase Design}
We first consider the binary phase-shift case, where $\mathcal{F}=\{0,\pi\}$. In this case, each beamforming entry reduces to a sign in $\{-1,+1\}$. Therefore, \eqref{eq:sum_inverse_gain} can be rewritten as follows:
\begin{equation}
\label{eq:binary_beamforming}
\min_{\mathbf{w}\in\{-1,+1\}^N} 
f(\mathbf{w})=\sum_{k=1}^K\frac{1}{\mathbf{w}^T\mathbf{R}_k\mathbf{w}}
\quad \text{s.t.} \quad w_1=1,
\end{equation}
where $\mathbf{R}_k=\Re\{\mathbf{h}_k\mathbf{h}_k^H\}$ and $w_n$ represents the $n$th element of ${\mathbf{w}}$. The constraint $w_1=1$ removes the phase-inversion symmetry and reduces the search space by half. To develop an efficient BB algorithm, we derive a lower bound for $f(\mathbf w)$. Let $\mathcal G$ and $\mathcal H$ denote the sets of fixed and free indices, respectively, and let $d=|\mathcal G|$. Define $\mathbf R\triangleq \sum_{k=1}^{K}\mathbf R_k$, and partition $\mathbf w$ and $\mathbf R$ according to $\mathcal G$ and $\mathcal H$ as follows:
\begin{equation}
\mathbf w=
\begin{bmatrix}
\mathbf w_{\mathcal G}\\
\mathbf w_{\mathcal H}
\end{bmatrix},
\qquad
\mathbf R=
\begin{bmatrix}
\mathbf R_{\mathcal {GG}} & \mathbf R_{\mathcal {GH}}\\
\mathbf R_{\mathcal {HG}} & \mathbf R_{\mathcal {HH}}
\end{bmatrix},
\end{equation}
where $\mathbf w_{\mathcal G}\in{\mathbb{R}}^{d\times1}$ and $\mathbf w_{\mathcal H}\in{\mathbb{R}}^{(N-d)\times1}$ collect the fixed and free entries of $\mathbf w$, respectively, and $\mathbf R_{\mathcal {GG}}\in{\mathbb{R}}^{d\times d}$, $\mathbf R_{\mathcal {GH}}\in{\mathbb{R}}^{d\times(N-d)}$, and $\mathbf R_{\mathcal {HH}}\in{\mathbb{R}}^{(N-d)\times(N-d)}$ denote the corresponding submatrices. The following lemma provides an upper bound on the quadratic form $\mathbf w^T\mathbf R\mathbf w$.
\vspace{-5pt}
\begin{lemma}\label{Lemma_UB_Binary}
For any feasible assignment of the free variables $\mathbf w_{\mathcal H}\in\{-1,+1\}^{N-d}$, the quadratic form $\mathbf w^T\mathbf R\mathbf w$ satisfies
\begin{equation}
\mathbf w^T\mathbf R\mathbf w \le \mathrm{UB}_{\mathrm{tot}},
\end{equation}
where
\begin{equation}
\mathrm{UB}_{\mathrm{tot}}
=
C_{\mathrm{tot}}
+
2\left\|\mathbf w_{\mathcal G}^T\mathbf R_{\mathcal {GH}}\right\|_1
+
(N-d)\lambda_{\max}\left(\mathbf R_{\mathcal {HH}}\right),
\tag{10}
\end{equation}
and where $C_{\rm tot}\triangleq\mathbf{w}_{\mathcal{G}}^T\mathbf{R}_{\mathcal{GG}}\mathbf{w}_{\mathcal{G}}$, and $\lambda_{\max}\left(\cdot\right)$ returns the principal eigenvalue.
\end{lemma}
\vspace{-5pt}
\begin{IEEEproof}
By block expansion, we have
\begin{align}
{\mathbf{w}}^T{\mathbf{R}}{\mathbf{w}}=\mathbf{w}_{\mathcal{G}}^T\mathbf{R}_{\mathcal{GG}}\mathbf{w}_{\mathcal{G}}+2\mathbf{w}_{\mathcal{G}}^T\mathbf{R}_{\mathcal{GH}}\mathbf{w}_{\mathcal{H}}
+\mathbf{w}_{\mathcal{H}}^T\mathbf{R}_{\mathcal{HH}}\mathbf{w}_{\mathcal{H}}\nonumber.
\end{align}
The first term is exactly $C_{\mathrm{tot}}$. For the cross term, since each entry of $\mathbf w_{\mathcal H}$ belongs to $\{-1,+1\}$, we obtain
\begin{equation}
\mathbf{w}_{\mathcal{G}}^T\mathbf{R}_{\mathcal{GH}}\mathbf{w}_{\mathcal{H}}
\le
\left\|\mathbf w_{\mathcal G}^T\mathbf R_{\mathcal {GH}}\right\|_1.
\end{equation}
For the last term, $\mathbf R_{\mathcal {HH}}$ is a principal submatrix of the positive semidefinite matrix $\mathbf R$, and is thus positive semidefinite. By the Rayleigh quotient bound,
\begin{equation}
\mathbf{w}_{\mathcal{H}}^T\mathbf{R}_{\mathcal{HH}}\mathbf{w}_{\mathcal{H}}
\le
\lambda_{\max}\left(\mathbf R_{\mathcal {HH}}\right)\|\mathbf w_{\mathcal H}\|_2^2.
\end{equation}
Since $\mathbf w_{\mathcal H}\in\{-1,+1\}^{N-d}$, we have $\|\mathbf w_{\mathcal H}\|_2^2=N-d$. Therefore, $\mathbf{w}_{\mathcal{H}}^T\mathbf{R}_{\mathcal{HH}}\mathbf{w}_{\mathcal{H}}\le (N-d)\lambda_{\max}\left(\mathbf R_{\mathcal {HH}}\right)$. Combining the above inequalities proves the lemma.
\end{IEEEproof}
Based on Lemma \ref{Lemma_UB_Binary}, we next derive a lower bound for the objective function.
\vspace{-5pt}
\begin{theorem}
For any feasible $\mathbf w\in\{-1,+1\}^N$ with $w_1=1$, the objective function satisfies $f(\mathbf w)\ge \mathrm{LB}\triangleq \frac{K^2}{\mathrm{UB}_{\mathrm{tot}}}$.
\end{theorem}
\vspace{-5pt}
\begin{IEEEproof}
By the arithmetic-harmonic mean inequality,
\begin{equation}
f(\mathbf w)=\sum_{k=1}^{K}\frac{1}{\mathbf w^T\mathbf R_k\mathbf w}
\ge
\frac{K^2}{\sum_{k=1}^{K}\mathbf w^T\mathbf R_k\mathbf w}
=
\frac{K^2}{\mathbf w^T\mathbf R\mathbf w}.
\end{equation}
By Lemma \ref{Lemma_UB_Binary}, we have $\mathbf w^T\mathbf R\mathbf w \le \mathrm{UB}_{\mathrm{tot}}$. Hence,
\begin{equation}
f(\mathbf w)\ge \frac{K^2}{\mathrm{UB}_{\mathrm{tot}}}=\mathrm{LB}.
\end{equation}
This completes the proof.
\end{IEEEproof}
During the BB procedure \cite{narendra1977branch}, any node with $\mathrm{LB}\ge f^*$, where $f^*$ denotes the current best objective value, is pruned. The overall procedure is summarized in Algorithm \ref{alg:BB_Binary}.

\begin{algorithm}[t]
\algsetup{linenosize=\tiny} \scriptsize
\caption{BB Algorithm for Binary Phase Beamforming}
\label{alg:BB_Binary}
\begin{algorithmic}[1]
\REQUIRE $\mathbf{R}$, $\{\mathbf{R}_k\}_{k=1}^K$
\ENSURE $\mathbf{w}^*$, $f^*$
\STATE $f^*\leftarrow\infty$, $\mathcal{Q}\leftarrow\{\text{root}\}$, $w_1=1$, $\mathcal{G}=\{1\}$, $\mathcal{H}=\{2,\dots,N\}$
\WHILE{$\mathcal{Q}\neq\emptyset$}
\STATE Extract node ($\mathcal{G}$, $\mathcal{H}$); compute ${\rm{UB}}_{\rm tot}$ and ${\rm{LB}}=\frac{K^2}{{\rm{UB}}_{\rm tot}}$
\IF{${\rm{LB}}<f^*$}
\IF{$|\mathcal{G}|=N$}
\STATE $f(\mathbf{w})=\sum_k(\mathbf{w}^T\mathbf{R}_k\mathbf{w})^{-1}$; update $f^*,\mathbf{w}^*$ if smaller
\ELSE
\STATE Pick $n\in\mathcal{H}$; branch $w_n\in\{+1,-1\}$; push children
\ENDIF\ENDIF
\ENDWHILE
\RETURN $\mathbf{w}^*$, $f^*$
\end{algorithmic}
\end{algorithm}

\subsection{$M$-ary-Phase Design}
While binary beamforming provides a low-cost baseline, many phased arrays employ higher-resolution phase shifters to improve beamforming gain. We next extend the BB framework to the general $M$-ary discrete phase case. The resulting problem can be expressed as follows:
\begin{equation}
\label{eq:mary_problem}
\begin{aligned}
\min_{\mathbf{w}} \quad & f(\mathbf{w}) = \sum_{k=1}^{K} \frac{1}{|\mathbf{h}_k^H \mathbf{w}|^2} \\
\text{s.t.} \quad & w_1 = 1, w_n \in \mathcal{W},  n = 2, \dots, N,
\end{aligned}
\end{equation}
where $\mathcal{W} \triangleq \{e^{j2\pi m/M}\}_{m=0}^{M-1}$ denotes the $M$-ary discrete phase set. Compared with the binary case, the search space grows to $M^{N-1}$. To maintain pruning efficiency, we develop a dual-layer lower bound that combines a user-wise bound with an aggregate system-level bound.
\subsubsection{Node decomposition}
Consider a node $S_d$ at depth $d$. Let $\mathcal F=\{1,\ldots,d\}$ and $\mathcal U=\{d+1,\ldots,N\}$ denote the sets of fixed and unfixed indices, respectively. Accordingly, we partition the beamforming vector as $\mathbf w=
\left[\begin{smallmatrix}
\mathbf w_{\mathcal F}\\
\mathbf w_{\mathcal U}
\end{smallmatrix}\right]$. For user $k$, the effective channel response can be decomposed as follows:
\begin{equation}
\mathbf h_k^H\mathbf w
=
\sum_{n=1}^{d} h_{kn}^* w_n
+
\sum_{n=d+1}^{N} h_{kn}^* w_n
\triangleq A_k+B_k,
\tag{17}
\end{equation}
where $A_k$ is determined by the fixed entries and $B_k$ collects the contribution of the unfixed entries.
\subsubsection{Layer 1: Individual lower bound}
By the triangle inequality, $|\mathbf h_k^H\mathbf w|
\le
|A_k|+\sum_{n=d+1}^{N}|h_{kn}|$. Define $|y_k|_{\max}\triangleq |A_k|+\sum_{n=d+1}^{N}|h_{kn}|$. Then, for any feasible completion of the unfixed phases, we have
\begin{equation}
\frac{1}{|\mathbf h_k^H\mathbf w|^2}\ge \frac{1}{|y_k|_{\max}^2}.
\end{equation}
This yields the individual lower bound as follows:
\begin{equation}
{\rm{LB}}_{\mathrm{indiv}}(S_d)\triangleq \sum_{k=1}^{K}\frac{1}{|y_k|_{\max}^2}\le f(\mathbf w).
\tag{18}
\end{equation}
\subsubsection{Layer 2: Aggregate lower bound}
To further exploit the coupling among users, define $\mathbf R_{\circ}\triangleq \sum_{k=1}^{K}\mathbf h_k\mathbf h_k^H$. Then, $\sum_{k=1}^{K}|\mathbf h_k^H\mathbf w|^2
=
\mathbf w^H\mathbf R_{\circ}\mathbf w$. Partition $\mathbf R_{\circ}$ according to $\mathcal F$ and $\mathcal U$ as follows:
\begin{equation}
\mathbf R_{\circ}
=
\begin{bmatrix}
\mathbf R_{\circ,\mathcal {FF}} & \mathbf R_{\circ,\mathcal {FU}}\\
\mathbf R_{\circ,\mathcal {UF}} & \mathbf R_{\circ,\mathcal {UU}}
\end{bmatrix}.
\end{equation}
The following lemma provides an upper bound on $\mathbf w^H\mathbf R_{\circ}\mathbf w$.
\vspace{-5pt}
\begin{lemma}
For any feasible assignment of the unfixed variables $\mathbf w_{\mathcal U}\in\mathcal W^{N-d}$, the quadratic form $\mathbf w^H\mathbf R_{\circ}\mathbf w$ satisfies
\begin{equation}
\begin{split}
\mathbf w^H\mathbf R_{\circ}\mathbf w &\le {\rm{UB}}_{\mathrm{tot}}\triangleq
\mathbf w_{\mathcal F}^H\mathbf R_{\circ,\mathcal {FF}}\mathbf w_{\mathcal F}\\
&+
2\left\|\mathbf w_{\mathcal F}^H\mathbf R_{\circ,\mathcal {FU}}\right\|_1
+
(N-d)\lambda_{\max}(\mathbf R_{\circ,\mathcal {UU}}).
\end{split}
\end{equation}
\end{lemma}
\vspace{-5pt}
\begin{IEEEproof}
Similar to the proof of Lemma \ref{Lemma_UB_Binary}.
\end{IEEEproof}
By the arithmetic-harmonic mean inequality, we have
\begin{equation}
\sum_{k=1}^{K}\frac{1}{|\mathbf h_k^H\mathbf w|^2}
\ge
\frac{K^2}{\sum_{k=1}^{K}|\mathbf h_k^H\mathbf w|^2}
=
\frac{K^2}{\mathbf w^H\mathbf R_{\circ}\mathbf w}.
\end{equation}
Hence, the aggregate lower bound is
\begin{equation}
{\rm{LB}}_{\mathrm{agg}}(S_d)\triangleq \frac{K^2}{{\rm{UB}}_{\mathrm{tot}}}\le f(\mathbf w).
\end{equation}
\subsubsection{Combined pruning criterion}
Since both ${\rm{LB}}_{\mathrm{indiv}}(S_d)$ and ${\rm{LB}}_{\mathrm{agg}}(S_d)$ are valid lower bounds of the objective function $f(\mathbf{w})$, we combine them as follows:
\begin{equation}\label{LB_Mary}
{\rm{LB}}(S_d)\triangleq \max\big({\rm{LB}}_{\mathrm{indiv}}(S_d),{\rm{LB}}_{\mathrm{agg}}(S_d)\big)
\le f(\mathbf w).
\end{equation}
By \eqref{LB_Mary}, the proposed global discrete BB algorithm prunes only nodes that cannot contain the global optimum. Therefore, it converges to the globally optimal solution. The complete procedure is summarized in Algorithm \ref{alg:GDBB_compact}.

\begin{algorithm}[t]
\algsetup{linenosize=\tiny} \scriptsize
\caption{BB Algorithm for $M$-ary Phase Beamforming}
\label{alg:GDBB_compact}
\begin{algorithmic}[1]
\REQUIRE $\{\mathbf{h}_k\}$, $\mathcal{W}$, $N$
\STATE $f^*\leftarrow\infty$, $\mathcal{Q}\leftarrow\{\mathcal{S}_1\}$ with $w_1=1$
\WHILE{$\mathcal{Q}\neq\emptyset$}
\STATE Pop $\mathcal{S}_d$ (min ${\rm{LB}}$)
\IF{${\rm{LB}} < f^*$}
\IF{$d = N$}
\STATE $f^*\leftarrow\min\big(f^*, f(\mathbf{w})\big)$
\ELSE
\FOR{$w_{d+1}\in\mathcal{W}$}
\STATE Push $\mathcal{S}_{d+1}$ if ${\rm{LB}} < f^*$
\ENDFOR
\ENDIF
\ENDIF
\ENDWHILE
\RETURN $\mathbf{w}^*, f^*$
\end{algorithmic}
\end{algorithm}

\subsection{Continuous-Phase Design}
While discrete-phase beamforming is attractive for its low cost, the continuous-phase design provides a useful performance upper bound. In the sequel, we relax the phase constraint to $\theta_n\in[0,2\pi)$ and develop a spatial BB (SBB) algorithm that attains the global optimum within a prescribed tolerance $\epsilon>0$.
\subsubsection{Epigraph reformulation and anchor-based lifting}
To handle the non-convex unit-modulus constraint $|w_n|=1$ and the sum-of-inverses objective, we first rewrite problem \eqref{eq:sum_inverse_gain} in epigraph form by introducing auxiliary variables $\{t_k\}$:
\begin{equation}\label{Contnuous_Problem}
\begin{aligned}
\min_{\mathbf w,\{t_k\}} \quad & \sum_{k=1}^{K} t_k \\
\text{s.t.} \quad & t_k\,\mathbf w^H \mathbf H_k \mathbf w \ge 1,\quad k=1,\ldots,K, \\
& |w_n|=1,\quad n=1,\ldots,N,
\end{aligned}
\end{equation}
where $\mathbf H_k=\mathbf h_k\mathbf h_k^H \succeq 0$. Next, define the augmented vector $\tilde{\mathbf w}=\left[\begin{smallmatrix}1\\ \mathbf w\end{smallmatrix}\right]\in\mathbb C^{N+1}$ and the lifted matrix $\tilde{\mathbf W}=\tilde{\mathbf w}\tilde{\mathbf w}^H \in \mathbb C^{(N+1)\times(N+1)}$. Then $\tilde{\mathbf W}\succeq 0$, $\tilde{\mathbf W}_{n,n}=1$, and the quadratic term becomes $\mathbf w^H \mathbf H_k \mathbf w=\operatorname{Tr}(\tilde{\mathbf H}_k \tilde{\mathbf W})$, where $\tilde{\mathbf H}_k=\left[
\begin{smallmatrix}
0 & \mathbf 0^T\\
\mathbf 0 & \mathbf H_k
\end{smallmatrix}\right]$. Here, $\tilde{\mathbf W}_{n,n'}$ denotes the $(n+1,n'+1)$th element of $\tilde{\mathbf W}$ for $n,n'\in\{0,\ldots,N\}$.

For each SBB node, define the phase box as follows:
\begin{equation}\label{Sector_Cons}
\mathcal R \triangleq \prod_{n=1}^{N}[\theta_n^L,\theta_n^U],
\end{equation}
where $\theta_n^L$ and $\theta_n^U$ are the lower and upper phase bounds of the $n$th entry. Let $\Delta\theta_n\triangleq \theta_n^U-\theta_n^L$ and $\phi_n^{\mathrm{mid}}\triangleq (\theta_n^L+\theta_n^U)/2$. We assume $\Delta\theta_n\le \pi$ for all $n$; otherwise, the interval is split before solving the node subproblem. The convex hull of the circular sector associated with $[\theta_n^L,\theta_n^U]$ is described by
\begin{equation}\label{Convex_Hull}
\Re\!\left\{\tilde{\mathbf W}_{n,0}e^{-j\phi_n^{\mathrm{mid}}}\right\}
\ge
\cos\!\left(\frac{\Delta\theta_n}{2}\right),
\quad n=1,\ldots,N.
\end{equation}

\begin{figure*}[htbp]
    \centering
    \begin{subfigure}[t]{0.25\textwidth}  
        \centering
        \includegraphics[width=\textwidth]{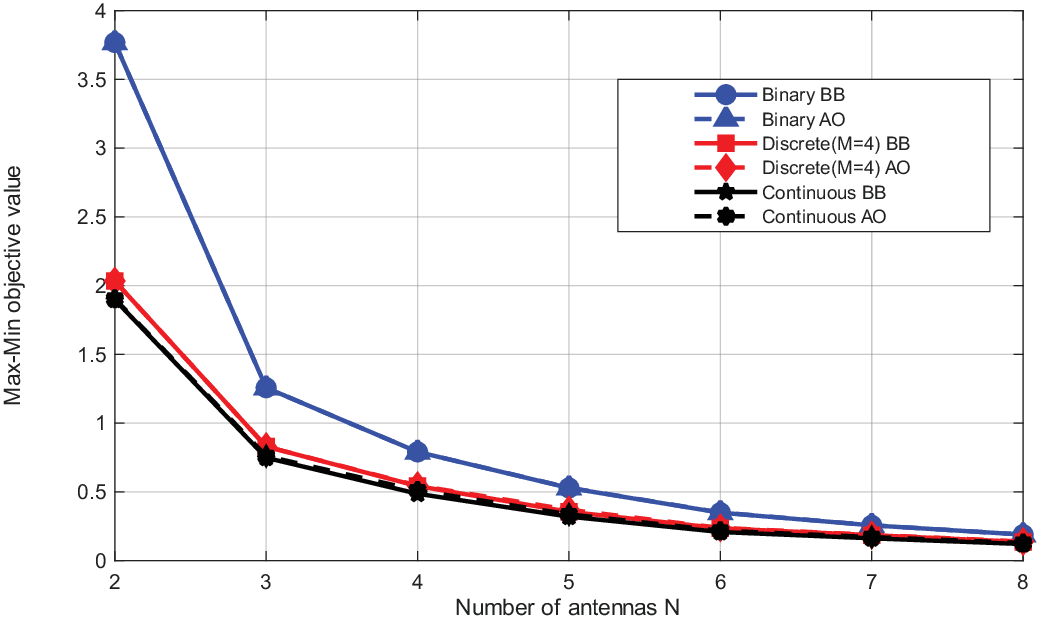}
        \caption{Antenna sweep ($K = 3$).}
        \label{fig:antenna_sweep_pdf}
    \end{subfigure}
    \begin{subfigure}[t]{0.25\textwidth}
        \centering
        \includegraphics[width=\textwidth]{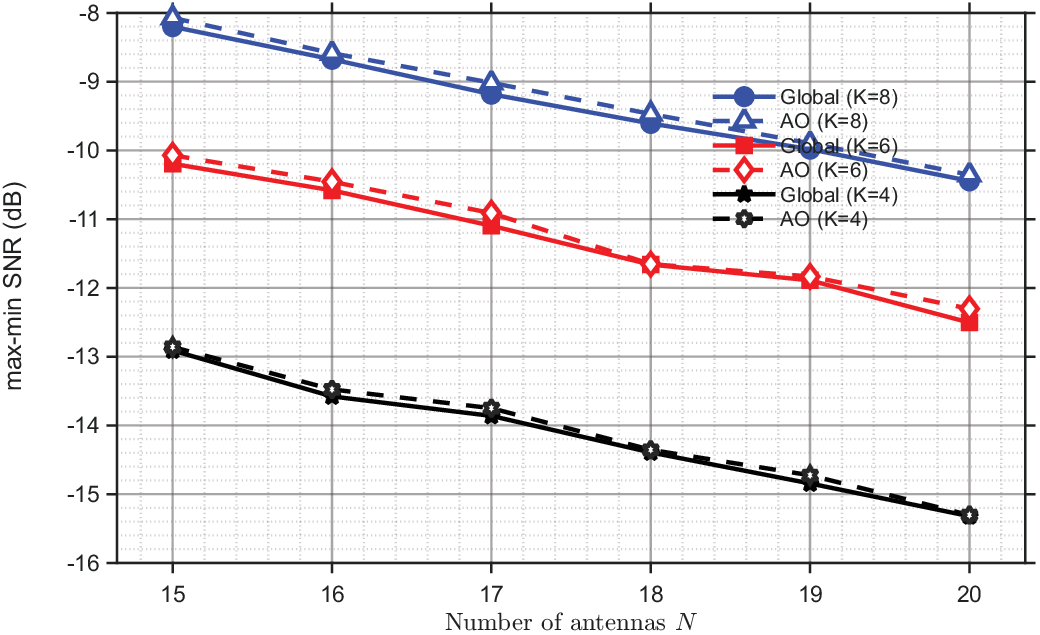}
        \caption{Binary phase ($M = 2$).}
        \label{fig:binary_phase_pdf}
    \end{subfigure}
    \begin{subfigure}[t]{0.25\textwidth}
        \centering
        \includegraphics[width=\textwidth]{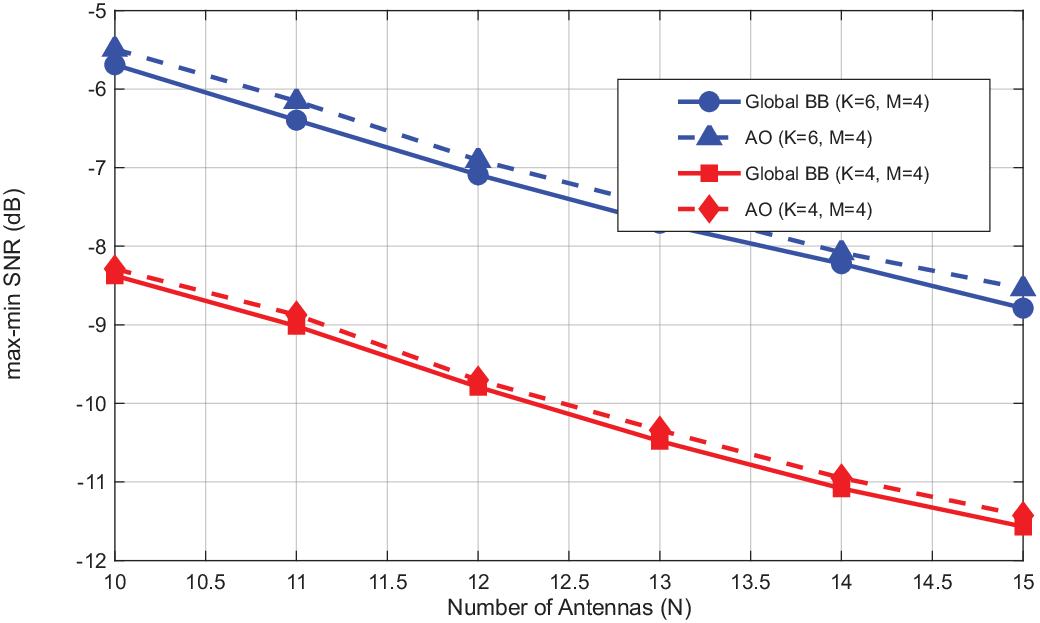}
        \caption{Discrete phase ($M = 4$).}
        \label{fig:discrete_phase_pdf}
    \end{subfigure}
    \begin{subfigure}[t]{0.25\textwidth}
        \centering
        \includegraphics[width=\textwidth]{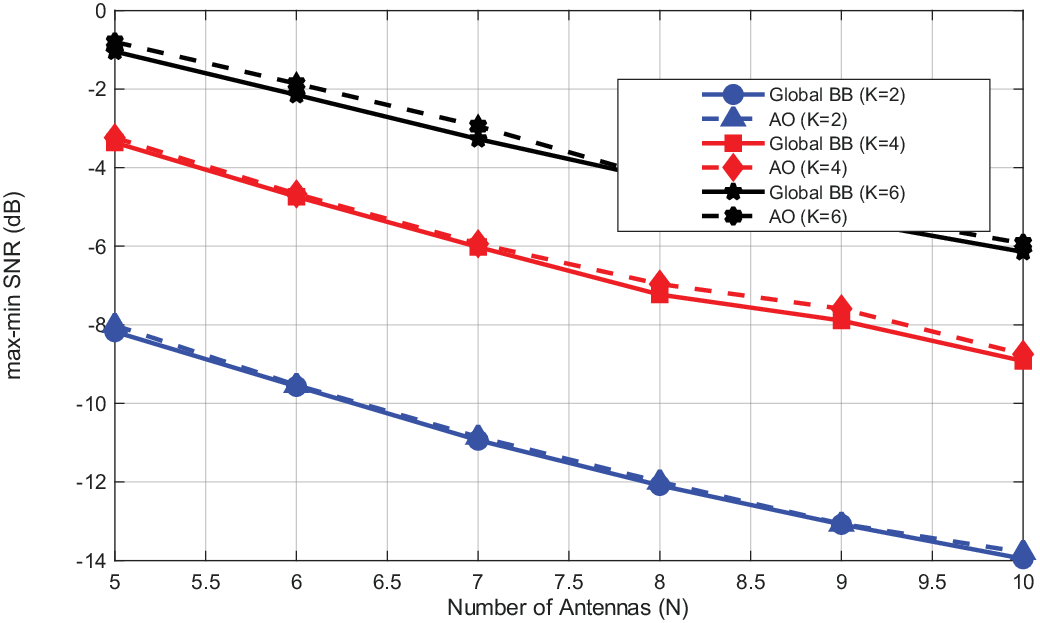}
        \caption{Continuous phase.}
        \label{fig:continuous_phase_pdf}
    \end{subfigure}
    \caption{Performance comparison of the proposed BB algorithms under various configurations: (a) objective value vs. $N$; (b) global vs. AO for binary phases; (c) global vs. AO for $M=4$ discrete phases; (d) global vs. AO for continuous phases.}
    \label{fig:combined_results_pdf}
\end{figure*}

\subsubsection{Convex relaxation over a subregion}
Over a given phase subregion $\mathcal R$, we solve the following semidefinite relaxation:
\begin{equation}\label{SBB_LB_Problem}
\begin{aligned}
\min_{\tilde{\mathbf W}\succeq 0,\{t_k\}} \quad & \sum_{k=1}^{K} t_k \\
\text{s.t.} \quad &
\begin{bmatrix}
t_k & 1\\
1 & \operatorname{Tr}(\tilde{\mathbf H}_k\tilde{\mathbf W})
\end{bmatrix}
\succeq 0,\quad k=1,\ldots,K, \\
& \Re\!\left\{\tilde{\mathbf W}_{n,0}e^{-j\phi_n^{\mathrm{mid}}}\right\}
\ge
\cos\!\left(\frac{\Delta\theta_n}{2}\right),
 n=1,\ldots,N, \\
& \tilde{\mathbf W}_{0,0}=1,\tilde{\mathbf W}_{n,n}=1,n=1,\ldots,N.
\end{aligned}
\end{equation}
Let ${\rm{LB}}_{\mathcal R}$ denote the optimal value of \eqref{SBB_LB_Problem}, and let $f_{\mathcal R}^{\star}$ denote the optimal value of the original problem \eqref{Contnuous_Problem} restricted to the subregion $\mathcal R$. Then ${\rm{LB}}_{\mathcal R}$ is a valid lower bound on $f_{\mathcal R}^{\star}$, which is detailed in the following theorem.
\vspace{-5pt}
\begin{theorem}\label{BB_LB_Continuous}
For any subregion $\mathcal R$, ${\rm{LB}}_{\mathcal R}\le f_{\mathcal R}^{\star}$. Moreover, let $\{\mathcal R^{(\ell)}\}_{\ell=1}^{\infty}$ be any nested sequence of subregions, where $\mathcal R^{(\ell)}=\prod_{n=1}^{N}[\theta_n^{L,(\ell)},\theta_n^{U,(\ell)}]$ and $\max_n \Delta\theta_n^{(\ell)}\to 0$. If the limiting phase vector is $\bm\theta^\star=[\theta_1^\star,\ldots,\theta_N^\star]^T$, the associated beamformer is $\mathbf w^\star=[e^{j\theta_1^\star},\ldots,e^{j\theta_N^\star}]^T$, and $\mathbf h_k^H\mathbf w^\star\neq 0$ for all $k$, then ${\rm{LB}}_{\mathcal R^{(\ell)}}\to f(\mathbf w^\star)$.
\end{theorem}
\vspace{-5pt}
\begin{IEEEproof}
Any feasible $\mathbf w\in\mathcal R$ induces a feasible lifted matrix $\tilde{\mathbf W}=\tilde{\mathbf w}\tilde{\mathbf w}^H$ for \eqref{SBB_LB_Problem}. Hence ${\rm{LB}}_{\mathcal R}\le f_{\mathcal R}^{\star}$. For the asymptotic claim, as $\max_n \Delta\theta_n^{(\ell)}\to 0$, the sector constraints force $\tilde {\mathbf{W}}_{n,0}\to e^{j\theta_n^\star}$. Together with $\tilde{\mathbf W}\succeq \mathbf 0$, $\tilde{\mathbf W}_{0,0}=1$, and $\tilde{\mathbf W}_{n,n}=1$, this implies that the limiting matrix is rank one and equals $[1,(\mathbf w^\star)^T]^H[1,(\mathbf w^\star)^T]$. Therefore, the relaxation becomes exact in the limit, i.e., ${\rm{LB}}_{\mathcal R^{(\ell)}}\to f(\mathbf w^\star)$.
\end{IEEEproof}
\subsubsection{Feasible upper bound}
Given the relaxed solution $\tilde{\mathbf W}^{\star}$ of \eqref{SBB_LB_Problem}, we construct two feasible beamformers $\bar{\mathbf w}\triangleq[\bar w_1,\ldots,\bar w_N]^T$ and $\check{\mathbf w}\triangleq[\check w_1,\ldots,\check w_N]^T$ as follows:
\begin{equation}
\bar\theta_n \triangleq ({\theta_n^L+\theta_n^U})/{2},\qquad \bar w_n \triangleq e^{j\bar\theta_n},
\end{equation}
and
\begin{equation}
\check\theta_n \triangleq \Pi_{[\theta_n^L,\theta_n^U]}\big(\arg(\tilde {\mathbf{W}}_{n0}^{\star})\big),\qquad \check w_n \triangleq e^{j\check\theta_n},
\end{equation}
where $\Pi_{[a,b]}(x)\triangleq \min\{\max\{x,a\},b\}$. Both $\bar{\mathbf w}$ and $\check{\mathbf w}$ are feasible for the original problem over $\mathcal R$. We therefore define
\begin{equation}
{\rm{UB}}_{\mathcal R}\triangleq \min\{f(\bar{\mathbf w}),\,f(\check{\mathbf w})\},
\end{equation}
which satisfies $f_{\mathcal R}^{\star}\le {\rm{UB}}_{\mathcal R}$. In particular, along any nested sequence with $\max_n \Delta\theta_n^{(\ell)}\to 0$, we have $\bar{\mathbf w}^{(\ell)}\to \mathbf w^\star$, and thus ${\rm{UB}}_{\mathcal R^{(\ell)}}\to f(\mathbf w^\star)$ by continuity of $f(\mathbf w)$. Combining this fact with Lemma \ref{BB_LB_Continuous} yields
\begin{equation}
{\rm{UB}}_{\mathcal R^{(\ell)}}-{\rm{LB}}_{\mathcal R^{(\ell)}}\to 0.
\tag{31}
\end{equation}
The resulting SBB procedure is summarized in Algorithm \ref{alg:SBB_compact}.

\begin{algorithm}[t]
\algsetup{linenosize=\tiny} \scriptsize
\caption{BB for Continuous Phase Beamforming}
\label{alg:SBB_compact}
\begin{algorithmic}[1]
\REQUIRE $[0,2\pi)^N$, $\epsilon$
\STATE ${\rm{UB}}\leftarrow+\infty$, $\mathcal{L}\leftarrow\{[0,2\pi)^N\}$
\WHILE{$\mathcal{L}\neq\emptyset$}
\STATE Pick $\mathcal{R}\in\mathcal{L}$, solve semidefinite programming $\Rightarrow ({\rm{LB}},\tilde{\mathbf{W}})$
\IF{${\rm{LB}} < {\rm{UB}}-\epsilon$}
\STATE ${\rm{UB}}\leftarrow\min\big({\rm{UB}},\,\min\big\{f(\bar{\mathbf w}),\,f(\check{\mathbf w})\big\}\big)$
\IF{${\rm{LB}} < {\rm{UB}}-\epsilon$}
\STATE Branch $\mathcal{R}$ along $\max \Delta\theta_n$
\ENDIF
\ENDIF
\ENDWHILE
\RETURN $\mathbf{w}^*$
\end{algorithmic}
\end{algorithm}

\section{Numerical Results}
In this section, we provide numerical results to evaluate the proposed BB algorithms. The channels are modeled as independent and identically distributed Rayleigh fading, i.e., each entry of the channel vector follows $\mathcal{CN}(0,1)$. The noise power is normalized to $\sigma^2=1$, and the total transmit power is set to $P=10$ dBm. All results are averaged over $10^3$ independent channel realizations.

{\figurename}~{\ref{fig:antenna_sweep_pdf}} shows the objective value $f(\mathbf w)$ versus the number of antennas $N$ under different phase resolutions. As $N$ increases, the objective value decreases, which indicates an improved max-min SNR due to the larger beamforming gain. Continuous-phase beamforming achieves the best performance and serves as a benchmark. The $M=4$ discrete-phase scheme closely approaches this benchmark, whereas the binary-phase case ($M=2$) suffers a clear performance loss.

We next examine the optimality of the proposed BB framework by comparing it with the AO-based method. As shown in {\figurename}~{\ref{fig:binary_phase_pdf}} and {\figurename}~{\ref{fig:discrete_phase_pdf}}, for binary phases and $M=4$ discrete phases, respectively, the proposed BB algorithms consistently achieve smaller objective values than the AO method over different values of $N$ and $K$. At the same time, the performance gap is generally small, which indicates that the AO method already provides a near-optimal solution in many cases and is therefore a promising low-complexity candidate in practice. In contrast, the proposed BB framework attains the global optimum through systematic branching and pruning, and thus serves as a rigorous benchmark for quantifying the optimality gap of practical designs.

{\figurename}~{\ref{fig:continuous_phase_pdf}} reports the results for continuous-phase optimization. The SBB algorithm consistently outperforms the AO method and is particularly advantageous in the small-$N$ regime. As $N$ increases, the performance gap gradually decreases, which again suggests that the AO method becomes increasingly competitive in the large-array regime. Nevertheless, SBB still achieves the best performance and provides a reliable benchmark for evaluating low-complexity continuous-phase designs in hardware-constrained multi-antenna systems.

\section{Conclusion}
This letter studied the joint design of analog beamforming and power allocation for a multiuser TDMA system. We developed globally optimal BB algorithms for binary, $M$-ary discrete, and continuous phase shifts. The proposed framework also offers useful insights for future hybrid beamforming design under practical hardware constraints.
\bibliographystyle{IEEEtran}
\bibliography{mybib}
\end{document}